\documentclass{endm}
\usepackage{endmmacro,amsmath,mathrsfs}
\usepackage{graphicx}

\newcommand{\argmax}{\operatornamewithlimits{argmax}}
\def \ie {i.e.~}
\def \eg {e.g.~}
\def \pm {p.m.~}
\def \NP {\mathcal{NP}}
\def \K {\mathscr{K}}
\def \PPP {\mathscr{P}}

\begin{document}

\begin{verbatim}\end{verbatim}\vspace{2.5cm}

\begin{frontmatter}

\title{Cross-identification of stellar catalogs with multiple stars: Complexity and Resolution}

\author{Daniel Sever\'in\thanksref{ALL}\thanksref{email}}
\address{Depto. de Matem\'atica (FCEIA), Universidad Nacional de Rosario, Argentina\\
   CONICET, Argentina}
	\thanks[ALL]{Partially supported by grants PICT-2016-0410 (ANPCyT) and PID ING538 (UNR).}
	 \thanks[email]{Email: \href{mailto:daniel@fceia.unr.edu.ar} {\texttt{\normalshape daniel@fceia.unr.edu.ar}}}

\begin{abstract}
In this work, I present an optimization problem which consists of assigning entries of a stellar catalog to multiple entries of
another stellar catalog such that the probability of such assignment is maximum.
I show a way of modeling it as a Maximum Weighted Stable Set Problem which is further used to solve a real astronomical instance
and I partially characterize the forbidden subgraphs of the resulting family of graphs given by that reduction.
Finally, I prove that the problem is $\NP$-Hard.
\end{abstract}


\end{frontmatter}

\section{Introduction}
\label{sec:intro}

In the science of astronomy, it is common to record the position and other physical quantities of stellar objects
in astronomical catalogs. They are of extreme importance for various disciplines, such as navigation, space research and geodesy.
Naturally, in star catalogs, a single star has different designations according to the catalog being used that uniquely identifies it.
Suppose that $A$ and $B$ are star catalogs, and $id_{A}$, $id_{B}$ are the designations of the same star in $A$ and $B$ respectively.
It is often necessary to know $id_{B}$ given $id_{A}$.
This kind of cross-identification can be performed by software tools available on Internet, such as
Xmatch\footnote{\url{http://matthiaslee.github.io/Xmatch}}
or the web-based CDS X-Match Service\footnote{\url{http://cdsxmatch.u-strasbg.fr/xmatch}}, which usually use heuristic algorithms.
It was not until recently, however, that exact approaches began to be proposed. For instance, in \cite{journal:budavari}, a
cross-identification problem is solved through assignment problems via the Hungarian Algorithm.

The correspondence between two catalogs does not need to be one-to-one. Some stars appearing as single ones in one catalog could
correspond to multiple stars in the other.
Although some catalogs, such as SAO and PPM, inform whether a certain star is double or not,
available cross-matching tools do not take into account this piece of information about the star.

Consider the following cross-identification problem.
Given two catalogs $A$ and $B$ covering the same region of the sky and being $B$ denser than $A$,
the problem consists of finding the ``most probable'' assignment such that every star $a$ is assigned up to $k_a$ stars of $B$,
where $k_a$ is the \emph{multiplicity} of $a$ informed by catalog $A$.

The original motivation to study this novel matching problem has arisen during a joint collaboration with astrophysicist
Diego Sevilla \cite{journal:mipaper} and whose objective has been the development of a new digital version of the \emph{Cordoba Durchmusterung},
a star catalog widely used in the twentieth century.

In this work, I describe an optimization problem which I call $\K$-\emph{Matching Problem} and I give a polynomial-time reduction to the
Maximum Weighted Stable Set Problem (MWSSP).
This reduction is further used for solving a real instance.
I also present an open question concerning the forbidden subgraphs of the family of graphs that arise in that reduction
and I identify two of the forbidden subgraphs.
Then, I prove that the $\K$-\emph{Matching Problem} is $\NP$-Hard for a given $\K \geq 2$.

%
%
\section{Problem description and resolution}
\label{sec:model}

Consider two star catalogs where each star is represented as elements of a set $A$ or $B$.
Let $n_A$ and $n_B$ be the cardinality of $A$ and $B$ respectively.

For a given entry $a \in A$, let $k_a$ be the multiplicity of $a$ in the first catalog.
That is, if $a$ represents a single star then $k_a = 1$, if $a$ represents a double one then $k_a = 2$, and so on.
Also, let $\K$ be the largest multiplicity. 

The resolution of our problem is divided in two phases:
\begin{itemize}
\item Phase 1: From the astrometric and photometric data available from catalogs, generate an instance of the $\K$-Matching Problem.
\item Phase 2: Reduce that instance to an instance of the MWSSP and solve it.
\end{itemize}
The first phase depends on the structure of both catalogs and involves criteria in the field of Astronomy, which can be separated from the
mathematical description of the problem. For that reason, it will be discussed in an Online
Appendix\footnote{\url{http://fceia.unr.edu.ar/~daniel/CD/new/onlineapp.pdf}}.
In this section, only the second phase is addressed.

During the first phase, \emph{candidates sets of stars} $P_a \subset \PPP(B)$ are generated for each $a \in A$.
For instance, the set $P_a = \{ \emptyset, \{b_1\}, \{b_2\}, \{b_1, b_3\}\}$ indicates that $a$ can be assigned to $b_1$, $b_2$, the
pair $\{b_1, b_3\}$ or no one (indicated by the presence of $\emptyset$) with positive probability.
Naturally, every $j \in P_a$ must satisfy $|j| \leq k_a$.
For a given star $a \in A$ and a set $j \in P_a$, denote the event that ``$a$ corresponds to $j$'' by $a \rightarrow j$ and
its probability by $p(a \rightarrow j)$, which is computed during the first phase. Also,
$\sum_{j \in P_a} p(a \rightarrow j) = 1$.

An \emph{assignment} $f : A \rightarrow \PPP(B)$ is \emph{valid} when it satisfies $f(a) \in P_a$ for all $a \in A$, and for any
$a_1, a_2 \in A$ such that $a_1 \neq a_2$, 
then $f(a_1) \cap f(a_2) = \emptyset$,
\ie candidates of $B$ assigned to $a_1$ and $a_2$ must not share common stars. Let $\mathcal{F}$ be the space of valid assignments.
Each $f \in \mathcal{F}$ has a corresponding probability $p(f) = p(a_1 \rightarrow f(a_1), a_2 \rightarrow f(a_2),\ldots)$.
We are interested in finding the most probable assignment:
$f^* \in \argmax_{f \in \mathcal{F}} p(f)$.
Since the number of assignments is exponential, it makes little sense to perform the computation of the real probability of each one.
Thus, let us make a simplification at this point by supposing the following assumption:
\begin{center}
\emph{for all $f \in \mathcal{F}$ and $a, a' \in A$ such that $a \neq a'$, events
$a \rightarrow f(a)$ and $a' \rightarrow f(a')$ are independent each other}.
\end{center}

Let $\overline{p}(f) = \prod_{a \in A} p(a \rightarrow f(a))$. If the previous assumption holds, we would have $\overline{p}(f) = p(f)$.
Although it usually does not hold, the assignment $f$ that maximizes $\overline{p}(f)$ is enough good for practical purposes.
Denote $w_{aj} = -ln(p(a \rightarrow j))$ for $a \in A$ and $j \in P_a$, and let $w(f) = \sum_{a \in A} w_{a f(a)}$.
It is easy to see that an optimal assignment $f$ can be found by minimizing $w(f)$, which is linear. The problem is defined as follows:

\newpage

\noindent \underline{$\K$-Matching Problem}\\
\noindent \emph{INSTANCE}:\\
\indent \indent $n_A, n_B \in \mathbb{Z}_+$;\\
\indent \indent $A, B$ such that $|A| = n_A$, $|B| = n_B$;\\
\indent \indent $P_a \subset \PPP(B)$ such that $|j| \leq \K$ for all $j \in P_a$, for all $a \in A$;\\
\indent \indent $w_{aj} \in \mathbb{R}_+$ for all $j \in P_a$ such that $\sum_{j \in P_a} e^{-w_{aj}} = 1$, for all $a \in A$.\\
\noindent \emph{OBJECTIVE}: Obtain a valid assigment $f$ such that $w(f)$ is minimum.

\medskip

Below, I show that this problem can be polynomially transformed to the MWSSP.
Recall that, given a graph $G = (V,E)$ and weights $z \in \mathbb{R}_+^V$, MWSPP consists of finding a stable set $S \subset V$ of $G$
such that $z(S) = \sum_{v \in S} z_v$ is maximum.
Let $G=(V,E)$ be the graph such that $V = \{v_{aj} : a \in A,~ j \in P_a\}$,
\begin{multline*}
E = \{ (v_{aj}, v_{aj'}) : a \in A,~~ j,j' \in P_a,~~ j \neq j'\}~ \cup \\
\{ (v_{aj}, v_{a'j'}) : a,a' \in A,~~ a \neq a',~~ j \in P_a,~~ j' \in P_{a'},~~ j \cap j' \neq \emptyset\},
\end{multline*}
and consider weights $z_{aj} = M - w_{aj}$ where $M = \sum_{a \in V} \sum_{j \in P_a} w_{aj}$.
\begin{theorem} \label{reducstableset}
Let $S$ be an optimal stable set of the MWSSP.
The $\K$-Matching Problem is feasible if and only if $z(S) > M.(n_A-1)$ and, in that case, $f(a) = j$ for all $v_{aj} \in S$ is an
optimal assignment of the $\K$-Matching Problem.
\end{theorem}
\begin{proof}
If the $\K$-Matching Problem is feasible, there exists a valid assignment $\hat{f}$. Let $\hat{S} \subset V$ such that
$v_{aj} \in \hat{S}$ if and only if $\hat{f}(a) = j$. It is easy to see that $\hat{S}$ is a stable set of $G$ whose weight is greater
than $M.(n_A-1)$. Since $S$ is optimal, $z(S) \geq z(\hat{S}) > M.(n_A-1)$.

Conversely, assume that $z(S) > M.(n_A-1)$ and let $f(a) = j$ for all $v_{aj} \in S$.
First, let us prove that $f$ is a valid assignment. Suppose that there exists $a^* \in A$ such
that $v_{a^* j} \notin S$ for every $j$. Then, $z(S) \leq M.(n_A-1) - \sum_{v_{aj} \in S} w_{aj} \leq M.(n_A-1)$ which leads to a contradiction.
Then, $f$ is defined for all $a \in A$. In addition, if $v_{aj}, v_{aj'} \in S$ then $v_{aj} = v_{aj'}$ so $a$ is assigned to a unique
$j$. Furthermore, if $a,a' \in A$ and $b \in B$ such that $b \in j$ and $b \in j'$ for some $j \in P_a, j' \in P_{a'}$ then
$a = a'$ so $b$ is assigned to at most one star of $A$.
Now, let us prove that $f$ is optimal. Suppose that there exists a valid assignment $\hat{f}$ such that $w(\hat{f}) < w(f)$.
Again, let $\hat{S} \subset V$ such that $v_{aj} \in \hat{S}$ if and only if $\hat{f}(a) = j$.
It is easy to see that $\hat{S}$ is a stable set of $G$ whose weight is $M.n_A - w(\hat{f})$. Then, $z(\hat{S}) > M.n_A - w(f) = z(S)$,
which is absurd.
\end{proof}

Based on this reduction, an exact algorithm (which can be consulted in the Online Appendix) was implemented for solving instances of
the 2-Matching Problem. Then, a real catalog of 52313 stars (where 568 are doubles) was cross-identified against another of 83397 stars
in less than a minute of CPU time.
The algorithm, auxiliary files and the resulting catalog are available \cite{dataset:mendeley}.

Now, define $\mathcal{F}_{\K}$ as the family of graphs $G$ obtained by the previous reduction for any instance of the $\K$-Matching
Problem.
It is clearly that the 1-Matching Problem, \ie when no multiple stars are present in catalog $A$, can be trivially reduced to
the classic \emph{Maximum Weighted Matching Problem} (MWMP) over a bipartite graph $G_B$.
Indeed, our reduction gives the line graph of $G_B$. Therefore, $\mathcal{F}_1$ is the family of line graphs of
bipartite graphs. It is known from Graph Theory that, if $G$ belongs to such family, then the \emph{claw}, the \emph{diamond} and the
\emph{odd holes} are forbidden induced subgraphs of $G$.
This leads to the following:

\medskip

\noindent {\bf Open question.} Which are the forbidden induced subgraphs that characterize those graphs from $\mathcal{F}_{\K}$ for $\K \geq 2$?

\medskip

Although none of the mentioned subgraphs are forbidden for the case $\K \geq 2$ (they can be generated from instances of
the 2-Matching Problem as it is shown in Figure \ref{fig:forb2}), the claw can be generalized as follows:
\begin{figure}
	\centering
		\includegraphics[scale=0.4]{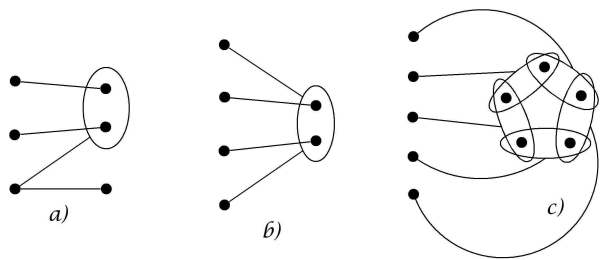}
	\caption{Instances for: a) claw, b) diamond, c) odd hole $C_5$}
	\label{fig:forb2}
\end{figure}
\begin{lemma} \label{lemita}
For $\K \geq 1$, let $G \in \mathcal{F}_{\K}$. Then, $G$ is $K_{1,\K+2}$-free.
\end{lemma}
\begin{proof}
Suppose that the star $K_{1,\K+2}$ is an induced subgraph of $G$. Let $v_{aj}$ be the central vertex of the star and
$v_{a_1,j_1}$, $v_{a_2,j_2}$, $\ldots$, $v_{a_{\K+2},j_{\K+2}}$ the remaining vertices. W.l.o.g., we can assume
that $a \neq a_1$, $a \neq a_2$, $\ldots$, $a \neq a_r$, $a = a_{r+1} = a_{r+2} = \ldots = a_{\K+2}$ for some
$r$. If $r \leq \K$, we would obtain that $a = a_{\K+1} = a_{\K+2}$ and then $v_{a,j_{\K+1}}$ and $v_{a,j_{\K+2}}$
would be adjacent which is absurd. Therefore, $r \geq \K+1$. Since $v_{aj}$ and $v_{a_i,j_i}$ are adjacent and
$a \neq a_i$ for all $1 \leq i \leq \K+1$, then $j \cap j_i \neq \emptyset$. On the other hand, $v_{a_i,j_i}$ and $v_{a_{i'},j_{i'}}$ are
not adjacent for all $1 \leq i < i' \leq \K+1$, then $j_i \cap j_{i'} = \emptyset$. Therefore, $j$ should have at least $\K+1$ elements
which leads to a contradiction.
\end{proof}
Another forbidden subgraph of the 2-Matching Problem is given as follows.
Let $G$ be the graph of Figure \ref{fig:forb3}(a).
Note that the instance of the 2-Matching Problem given in Figure \ref{fig:forb3}(b) corresponds to
the subgraph of $G$ induced by vertices $v_1,\ldots,v_7$. A drawback emerges when $v_8$ is considered. Hence, $G \notin \mathcal{F}_2$.
\begin{figure}
	\centering
		\includegraphics[scale=0.4]{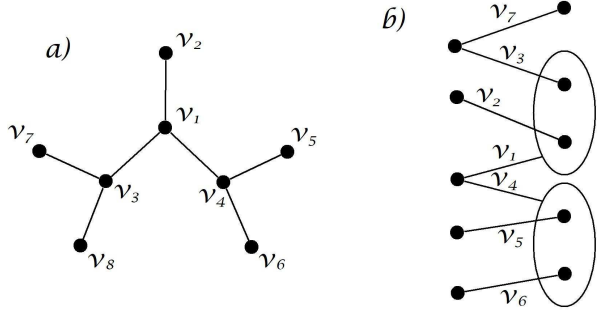}
	\caption{A graph not in $\mathcal{F}_2$: a) $G$, b) partial construction}
	\label{fig:forb3}
\end{figure}
\indent From the complexity point of view, the $\K$-Matching Problem for $\K = 1$ is polynomial due to the existence of efficient
algorithms for the MWMP such as the Hungarian Algorithm. When $\K = 2$, Lemma \ref{lemita} says that graphs from $\mathcal{F}_{\K}$
are $K_{1,4}$-free, and the MWSSP for $K_{1,4}$-free graphs is known to be $\NP$-Hard. Nevertheless, this does not mean that our
matching problem is hard since $\mathcal{F}_2$ has other forbidden subgraphs. Its complexity is addressed in the next section.

%
%
\section{Complexity of the problem}
\label{sec:complexity}

In this section, I prove that the $\K$-Matching Problem is $\NP$-hard for $\K \geq 2$. Even more, I consider a more restricted
problem where every star of $A$ has exactly multiplicity $\K$. The decision problem is as follows:

\medskip

{ \small
\noindent \underline{$\K$-Matching Decision Problem} ($\K$-\textbf{MDP})\\
\noindent \emph{INSTANCE}:~
   $n_A, n_B \in \mathbb{Z}_+$;~
   $A, B$ such that $|A| = n_A$, $|B| = n_B$;~
   $P_a \subset \PPP(B)$ such that $|j| = \K$ for all $j \in P_a$, for all $a \in A$;~
   $w_{aj} \in \mathbb{R}_+$ for all $j \in P_a$ such that $\sum_{j \in P_a} e^{-w_{aj}} = 1$, for all $a \in A$;~
   $t \in \mathbb{R}$.\\
\noindent \emph{QUESTION}: Is there a valid assignment $f$ such that $w(f) \leq t$ ?
}

\medskip

Let us first introduce two auxiliary problems.
Given $n \in \mathbb{Z}_+$, let
$\mathcal{P}$ and $\mathcal{Q}$ be disjoint sets such that $|\mathcal{P}|=|\mathcal{Q}|=n$. A \emph{perfect matching} (\pm for short) is
a set $M \subset \mathcal{P} \times \mathcal{Q}$ such that $|M|=n$ and every element of $\mathcal{P} \cup \mathcal{Q}$ occurs in exactly one
pair of $M$. The first, which is $\NP$-complete \cite{journal:frieze}, is defined below:

\medskip

{ \small
\noindent \underline{Disjoint Matchings} (\textbf{DM})\\
\noindent \emph{INSTANCE}:~
   $n \in \mathbb{Z}_+$;~
   disjoint sets $\mathcal{P}$, $\mathcal{Q}$ such that $|\mathcal{P}|=|\mathcal{Q}|=n$;~
   $\mathcal{A}_1,\mathcal{A}_2 \subset \mathcal{P} \times \mathcal{Q}$.\\
\noindent \emph{QUESTION}: Are there \pm $M_1 \subset \mathcal{A}_1, M_2 \subset \mathcal{A}_2$ such that
$M_1 \cap M_2 = \emptyset$ ?
}

\medskip

The second auxiliary problem is given below. It differs from the 2-Matching Decision Problem in that values $w_{aj}$ do not come from probabilities:

\medskip

{ \small
\noindent \underline{2-Matching Decision Problem with Arbitrary Weights} (2-\textbf{MDPAW})\\
\noindent \emph{INSTANCE}:~
   $n_A, n_B \in \mathbb{Z}_+$;~
   sets $A, B$ such that $|A| = n_A$ and $|B| = n_B$;~
   $P_a \subset \PPP(B)$ such that $|j| = 2$ for all $j \in P_a$, $a \in A$;~
   $w_{aj} \in \mathbb{R}_+$ for all $j \in P_a$, $a \in A$;~
   $t \in \mathbb{R}$.\\
\noindent \emph{QUESTION}: Is there a valid assignment $f$ such that $w(f) \leq t$ ?
}

\begin{lemma} \label{lemacomplejidad}
2-MDPAW is $\NP$-complete.
\end{lemma}
\begin{proof}
First of all, it clearly is $\NP$. Below, a polynomial transformation from DM is proposed.
Consider an instance $\mathcal{P} = \{p_1, \ldots, p_n\}$, $\mathcal{Q} = \{q_1, \ldots, q_n\}$,
$\mathcal{A}_1,\mathcal{A}_2 \subset \mathcal{P} \times \mathcal{Q}$ of DM. We construct an instance of 2-MDPAW as follows.
Let $A = \{a_{rs} : ~\textrm{$r$ and $s$ such that}~ (p_r,q_s) \in \mathcal{A}_1 \cup \mathcal{A}_2\}$ and
\begin{multline*}
B = \{ p^1_i, p^2_i : ~\textrm{$i$ such that}~ p_i \in \mathcal{P} \}~ \cup \\
    \{ q^1_i, q^2_i : ~\textrm{$i$ such that}~ q_i \in \mathcal{Q} \}~ \cup \\
    \{ z_{rs}, z'_{rs} : ~\textrm{$r$ and $s$}$ $\textrm{such that}~ a_{rs} \in A \}.
\end{multline*}
Hence, $n_A = |\mathcal{A}_1 \cup \mathcal{A}_2|$ and $n_B = 4n + 2|\mathcal{A}_1 \cup \mathcal{A}_2|$.
For every $a_{rs} \in A$, let $P_{a_{rs}} = \{\{p^i_r, q^i_s\} : ~\textrm{$r$, $s$ and $i$}$
$\textrm{such that}~(p_r,q_s) \in \mathcal{A}_i\} \cup \{ \{z_{rs}, z'_{rs}\} \}$.
For $a_{rs} \in A$ and $j \in P_{a_{rs}}$, let
\begin{equation*}
w_{a_{rs} j} = \begin{cases}
0, & j = \{p^i_r, q^i_s\} ~\textrm{for some $i$} ~\land~ (p_r,q_s) \in \mathcal{A}_1 \triangle \mathcal{A}_2, \\
1, & j = \{p^i_r, q^i_s\} ~\textrm{for some $i$} ~\land~ (p_r,q_s) \in \mathcal{A}_1 \cap \mathcal{A}_2, \\
1, & j = \{z_{rs}, z'_{rs}\} ~\land~ (p_r,q_s) \in \mathcal{A}_1 \triangle \mathcal{A}_2, \\
2, & j = \{z_{rs}, z'_{rs}\} ~\land~ (p_r,q_s) \in \mathcal{A}_1 \cap \mathcal{A}_2, \end{cases}
\end{equation*}
where $\triangle$ denotes the symmetric difference operator between sets. Finally, let $t = |\mathcal{A}_1| + |\mathcal{A}_2| - 2n$.

We prove that, given disjoint \pm $M_1 \subset \mathcal{A}_1, M_2 \subset \mathcal{A}_2$, there exists a valid assignment
$f$ such that $w(f) \leq t$. Consider $f(a_{rs}) = \{p^i_r, q^i_s\}$ when $(p_r, q_s) \in M_i$ for some $i \in \{1,2\}$, and
$f(a_{rs}) = \{z_{rs}, z'_{rs}\}$ otherwise. The validity of $f$ is straightforward. Also,
$w(f) = $ 
   $|(M_1 \cap \mathcal{A}_2) \cup (M_2 \cap \mathcal{A}_1)|$ $+$
	 $|(\mathcal{A}_1 \backslash (M_1 \cup \mathcal{A}_2)) \cup (\mathcal{A}_2 \backslash (M_2 \cup \mathcal{A}_1))|$ $+$
   $2|(\mathcal{A}_1 \cap \mathcal{A}_2) \backslash (M_1 \cup M_2)| =$
	 $|\mathcal{A}_1 \backslash M_1| + |\mathcal{A}_2 \backslash M_2| = t$.
Conversely, we prove that, for a given valid assignment $f$ such that $w(f) \leq t$, there exist disjoint \pm
$M_1 \subset \mathcal{A}_1, M_2 \subset \mathcal{A}_2$. Consider
$M_i = \{ (p_r, q_s) : ~\textrm{$r$ and}$ $\textrm{$s$ such that}~ f(a_{rs}) = \{p^i_r, q^i_s\}\}$ for all $i \in \{1,2\}$.
Since $f$ is a function, $M_1 \cap M_2 = \emptyset$. It is also straightforward that $M_i \subset \mathcal{A}_i$.
Now, suppose that there exists an element in $\mathcal{P} \cup \mathcal{Q}$ occurring in two pairs of $M_i$.
W.l.o.g., suppose $(p_1,q_1), (p_1,q_2) \in M_1$. Then, $f(a_{11}) \cap f(a_{12}) = \{p^1_1,q^1_1\} \cap \{p^1_1,q^1_2\} \neq \emptyset$
which is absurd. Therefore, every element in $\mathcal{P} \cup \mathcal{Q}$ occur at most once in any pair of $M_1$ and once in $M_2$.
It is easy to see that $|M_1| \leq n$ and $|M_2| \leq n$.
Suppose that there exists an element in $\mathcal{P} \cup \mathcal{Q}$ which does not occur in any pair of $M_i$.
Again, w.l.o.g., suppose that such element does not occur in $M_1$. Then, $|M_1| < n$ and
$w(f) = |\mathcal{A}_1 \backslash M_1| + |\mathcal{A}_2 \backslash M_2| > |\mathcal{A}_1| + |\mathcal{A}_2| - 2n = t$. Absurd!
Therefore, $M_1$ and $M_2$ are both \pm and $|M_1| = |M_2| = n$.
\end{proof}

\begin{theorem} \label{kmatchinghard}
$\K$-MDP is $\NP$-complete for all $\K \geq 2$.
\end{theorem}
\begin{proof}
We propose a polynomial transformation from 2-MDPAW. Consider an instance
$A = \{a_1, \ldots, a_{n_A}\}$, $B = \{b_1, \ldots, b_{n_B}\}$, $P_a$, $w_{aj}$, and $t$ of 2-MDPAW.
We construct an instance $A', B', P'_a, w'_{aj}, t'$ of $\K$-MDP as follows.
Let $A' = A \cup \{ \bar{a}_1, \ldots, \bar{a}_{n_A} \}$ and
$B' = B \cup \{ \tilde{b}_{jk} : j \in \bigcup_{a \in A} P_a,~ 3 \leq k \leq \K \} \cup \{ \bar{b}_{ak} : a \in A,~ 1 \leq k \leq \K \}$.
For all $a \in A$, let $P'_a = \{ j \cup \bigcup_{k=3}^{\K} \tilde{b}_{jk} : j \in P_a \} \cup \{ j'_a \}$
where $j'_a = \bigcup_{k=1}^{\K} \bar{b}_{ak}$ (if $\K = 2$, we just have
$P'_a = P_a \cup \{ \{\bar{b}_{a1}, \bar{b}_{a2}\} \}$).
Take an $a^* \in A$ that maximizes $p^* \doteq \sum_{j \in P_{a^*}} e^{-w_{a^*j}}$. Let $\beta > ln(p^*)$
and $w'_{aj} = w_{aj} + \beta$ for all $j \in P_a$, $a \in A$. Then, $\sum_{j \in P_a} e^{-w'_{aj}} < 1$.
Let $w'_{aj'_a} = -ln(1-\sum_{j \in P_a} e^{-w'_{aj}})$ for all $a \in A$. We obtain $\sum_{j \in P'_a} e^{-w'_{aj}} = 1$.
For all $i \in \{1,\ldots,n_A\}$, let $P'_{\bar{a}_i} = \{ j'_{\bar{a}_i} \}$ and $w'_{\bar{a}_i j'_{\bar{a}_i}} = 0$ where
$j'_{\bar{a}_i} = \bigcup_{k=1}^{\K} \bar{b}_{{a_i}k}$. Finally, let $t' = t + n_A \beta$.

Now we prove that there is an $f$ of 2-MDPAW such that $w(f) \leq t$ if and only if there is an
$f'$ of $\K$-MDP such that $w(f') \leq t'$. In order $f'$ to be valid, $f'(\bar{a}_i) = j'_{\bar{a}_i}$
for all $1 \leq i \leq n_A$. We propose $f'(a) = f(a)$ for all $a \in A$. Clearly, if $f$ is valid then $f'$ is valid too, and conversely.
Since $\sum_{a \in A' \backslash A} w'_{a f'(a)} = 0$, $w(f') = w(f) + n_A \beta$.
\end{proof}


\newpage

\section*{Online Appendix of ``Cross-identification of stellar catalogs with multiple stars: Complexity and Resolution''}

\subsection*{Example of a 2-Matching Problem}

Consider an instance of the 2-Matching Problem where $A = \{a_1, a_2, a_3, a_4\}$ and
$B = \{b_1, b_2, \ldots, b_6\}$. Here, $a_1, a_4$ are single stars and $a_2, a_3$ are double.
Suppose that the first phase yields the following sets:\\
\indent \indent $P_{a_1} = \{\{b_1\}, \{b_2\}, \{b_3\}\}$,\\
\indent \indent $P_{a_2} = \{\{b_2, b_3\}, \{b_4, b_5\}\}$,\\
\indent \indent $P_{a_3} = \{\{b_2, b_3\}, \{b_5, b_6\}\}$,\\
\indent \indent $P_{a_4} = \{\{b_6\}, \emptyset\}$.\\
A scheme that includes probabilities is displayed in Figure \ref{fig:ex1}(a).
Here, the optimal assignment is
$f^*(a_1) = \{b_1\}$,
$f^*(a_2) = \{b_4, b_5\}$,
$f^*(a_3) = \{b_2, b_3\}$,
$f^*(a_4) = \emptyset$ with probability $\overline{p}(f^*) = 0.1008$.

The reduction to the MWSSP gives $M = 8.3269$, weights\\
\indent \indent $w_{1\{1\}} = 7.1229$,~$w_{1\{2\}} = 7.6338$,~$w_{1\{3\}} = 6.7175$,\\
\indent \indent $w_{2\{2,3\}} = 7.1229$,~$w_{2\{4,5\}} = 7.9702$,\\
\indent \indent $w_{3\{2,3\}} = 8.1038$,~$w_{3\{5,6\}} = 6.7175$,\\
\indent \indent $w_{4\{6\}} = 7.4106$ and $w_{4\emptyset} = 7.8161$\\
(letters ``$a$'' and ``$b$'' are omitted for the sake of readability),
and the graph is shown in Figure \ref{fig:ex1}(b).

\subsection*{Example of the reduction of Lemma \ref{lemacomplejidad} and Theorem \ref{kmatchinghard}}

Consider the instance of DM given in Figure \ref{fig:proof1}(a) where $n = 2$ and
$|\mathcal{A}_1| = |\mathcal{A}_2| = 3$. The corresponding instance of 2-MDPAW is shown in Figure \ref{fig:proof1}(b)
where $t = 2$.

Also, for $\beta = 0.4$ and the given instance of 2-MDPAW, the corresponding instance of 3-MDP is shown in Figure \ref{fig:proof2}
where $t' = 3.6$. Vertices $\tilde{b}_{j3}$ for all $j \in \bigcup_{a \in A} P_a$ are displayed as unlabeled circles filled with white color.

\subsection*{Algorithm}

Here, an exact algorithm for the the $\K$-Matching Problem is proposed and the resolution of a cross-identification
between two catalogs based on real data is presented.

The algorithm is given below.
\begin{enumerate}
\item For each $a \in A$ such that $\emptyset \in P_a$, do the following. If there is an element $j \in P_a$ such that
$w_{aj} > w_{a\emptyset}$, remove $j$ from $P_a$ (if $f^*$ is an optimal assignment then $f^*(a) \neq j$ since $\emptyset$ is a better
choice than $j$).
\item Generate the graph $G$ as stated in Section \ref{sec:model}.
\item Find the connected components of $G$.
\item For each component $G'$ of $G$, solve the problem restricted to $G'$.
\end{enumerate}
Let $A'$ and $B'$ be the stars involved in a component $G'$ of $G$, \ie $A' = \{ a \in A : v_{aj} \in V(G') \}$ and
$B' = \{ b \in B : v_{aj} \in V(G'), b \in j \}$.
In the last step of our algorithm, three cases can be presented:
\begin{itemize}
\item \emph{Unique star}. If $A'=\{a\}$, then the solution is straightforward: $f^*(a)= argmin_{j \in P_a} w_{aj}$.
\item \emph{Only single stars}. If $|A'| \geq 2$ and $k_a = 1$ for all $a \in A$, then the problem restricted to $G'$ can be solved via the
Hungarian Algorithm in polynomial time. In that case, the instance of the MWMP is: a bipartite graph $G_B$ such that
$V(G_B) = A' \cup B' \cup \{ \emptyset_a : a \in A' ~\textrm{such that}~ \emptyset \in P_a \}$ and
$E(G_B) = \{ (a,b) : a \in A', \{b\} \in P_a \} \cup \{ (a,\emptyset_a) : a \in A', \emptyset \in P_a \}$,
weights $-w_{a\{b\}}$ for each edge $(a,b)$ and weights $-w_{a\emptyset}$ for each edge $(a,\emptyset_a)$.
\item \emph{Multiple stars}. If $|A'| \geq 2$ and there is $a \in A'$ such that $k_a \geq 2$, then it can be solved with an exact
algorithm for the MWSSP\footnote{See, for instance, S. Rebennack, M. Oswald, D. O. Theis, H. Seitz, G. Reinelt and P. M. Pardalos,
\emph{A Branch and Cut solver for the maximum stable set problem}, J. Comb. Optim. \textbf{21} (2011), 434--457.}.
In the case that such algorithm is not available, solving the following integer linear programming formulation is a reasonably fast
alternative:
\begin{align}
\min ~~\sum_{a \in A'} \sum_{j \in P_a} w_{aj} x_{aj} & & \notag \\
\textrm{subject to} & & \notag \\
 & \sum_{j \in P_a} x_{aj} = 1, & \forall~ a \in A' \label{constr:assign}\\
 & \sum_{a \in A'} ~~\sum_{j \in P_a : b \in j} x_{aj} \leq 1, & \forall~ b \in B'  \label{constr:otherside} \\
 & x_{aj} \in \{0, 1\}, & \forall~ a \in A',~ j \in P_a \notag
\end{align}
Constraints (\ref{constr:assign}) guarantee that each star of $A'$ must be assigned to exactly one element $j$ of $P_a$.
Constraints (\ref{constr:otherside}) forbid that each star of $B'$ be assigned to two or more stars of $A'$.
For the sake of readability, the latter constraints are presented for all $b \in B'$ but one have to keep in mind that some of them can be
removed if: (i) the constraint has just one variable in the left hand side, or (ii) it is repeated, \ie if, for some $b \in B'$, there
exists another $\tilde{b} \in B'$ such that $b$ and $\tilde{b}$ occur exactly in the same tuples of $\bigcup_{a \in A'} P_a$.
\end{itemize}

An instance of the 2-Matching Problem is obtained once the first phase is completed.
Table \ref{table:optim} reports some highlights about the optimization of that instance.

As we can see from the table, $G$ is highly decomposable and just 111 integer linear problems needs to be solved.
Moreover, these integer problems turned out to be very easy to solve since the solver did not branch (all of them were solved in the
root node). The hardest one has 339 variables and 118 constraints, and took 0.0015 seconds of CPU time.
The optimization was performed on a computer equipped with an Intel i7-7700 at 3.60 Ghz and GuRoBi 6.5.2 as the MIP solver.
The overall process took 41.6 seconds of CPU time.

\subsection*{Description of the first phase}

This section is devoted to present a summary on how to obtain a set of candidate stars for a given star of the former catalog and
the probabilities involved in them.
Recall that such computations heavily depends on structure and data availability of both catalogs as well as the underlying physical model
used to establish the relationship between them.
It is beyond the scope of this work to analyze such scenarios neither to give a formal treatment, so a simplified\footnote{Stars of both catalogs should not be
near the celestial poles in order to avoid certain distortions, and stars with high variability in its brightness should be avoided.
This can be done by pre-identifying them and remove them from both catalogs.} but reasonable model is considered, which is enough for
presenting our approach\footnote{A more robust and general probabilistic model is discussed in T. Budav\'ari and A. S. Szalay,
\emph{Probabilistic Cross-Identification of Astronomical Sources}, Astrophys. J. \textbf{679} (2008), 301--309.}.

Consider catalogs $A$ and $B$, and let $A_2$ be the set of stars from catalog $A$ marked as ``double''.
Our goal is to propose an instance of the 2-Matching Problem.\\

Let us first present some basic elements of Positional Astronomy.
Usually, position is given in a well established reference frame where two spherical coordinates are used:
\emph{right ascension} denoted by $\alpha$ and \emph{declination} denoted by $\delta$, similar to longitude and latitude coordinates on Earth.
In fact, a pair $(\alpha, \delta)$ represents a point in the unit sphere. For a given two points $p_1, p_2$, denote its \emph{angular
distance} by $\theta(p_1,p_2)$.
A known property is that, if points $p_1, p_2$ have the same right ascension, $\theta(p_1,p_2)$ is given by the difference in
its declinations. However, if $p_1, p_2$ have the same declination, $\theta(p_1,p_2)$ depends on the difference in right ascensions and
the cosine of the declination of both points. For this reason, it is convenient to work with the quantity $\alpha^* = \alpha.cos(\delta)$
instead of $\alpha$ directly.

Catalogs usually give the right ascension $\alpha$, declination $\delta$ and \emph{visual magnitude} $m$ (a measure of brightness) of
each star. These parameters are modeled as a multivariate normal distribution.
However, in several catalogs, each parameter is considered independent from each other.
Therefore, for a given star we have
$\overline{\alpha^*} \sim \mathcal{N}(\alpha^*, \sigma_{\alpha^*}^2)$,
$\overline{\delta} \sim \mathcal{N}(\delta, \sigma_{\delta}^2)$,
$\overline{m} \sim \mathcal{N}(m, \sigma_{m}^2)$,
where $\alpha^*$, $\delta$ and $m$ are the expected values of the parameters and $\sigma_{\alpha^*}$, $\sigma_{\delta}$ and $\sigma_{m}$
its standard errors.

Positions provided in a catalog are valid for a certain \emph{epoch}, which is a specific moment in time. However, there exist
transformations for translating positions from one epoch to other such as precession and nutation. In addition, stars have its own
apparent motion across the sky denominated \emph{proper motion}. Some catalogs also provide additional coefficients for computing the
correction in proper motion. These coefficients have its own standard errors. Therefore, it is possible to compute the positions and its
uncertainties of a star for a new epoch by means of the mentioned transformations and the propagation of the error\footnote{Details of
these transformations are treated in J. Kovalevsky and P. K. Seidelmann, \emph{Fundamentals of Astrometry},
Cambridge University Press, UK, 2004.}.
This is the case of the catalog PPMX\footnote{See S. Roeser, E. Schilbach, H. Schwan, N. V. Kharchenko, A. E. Piskunov and R.-D. Scholz,
\emph{PPM-Extended (PPMX), a catalogue of positions and proper motions}, Astron. Astrophys. \textbf{488} (2008), 401--408.}
where position for epoch $J2000.0$, brightness, proper motions and its uncertainties are available, among others parameters.

Naturally, older catalogs handle less information. For instance, the Cordoba Durchmusterung (CD) does not report standard errors
for each star, but a mean standard error over several stars from the same region of the sky\footnote{See pages XXIX-XXX of J. M. Thome,
\emph{Cordoba Durchmusterung (-22$^\circ$ to -32$^\circ$)}, Resultados del Observatorio Nacional Argentino \textbf{16} (1892).},
\eg for stars whose declinations are between $-22^\circ$ and $-32^\circ$ we have $\sigma_{\alpha^*} = 9.3~\textrm{arcsec}$ and
$\sigma_{\delta} = 20.5~\textrm{arcsec}$.\\

Some extra parameters $(\rho_{max}, d_{sep}, \sigma_{d_{sep}}, \sigma_{mag}, p_{\emptyset}, p_{sgl})$ must be determined before
performing the cross-identification.
Therefore, the input of our problem consists of catalogs $A$, $B$ and these extra parameters.
They will be introduced thoughout this section.\\

\noindent \emph{Treatment of single stars}.
Let $a \in A \backslash A_2$ and $b \in B$. 
Observe that, if $a$ and $b$ are far from each other, it makes little sense that both represent the same star. Usually, a criterion based on
the angular distance between them can be used to keep those ``close'' pairs.
Consider a candidate for $a$ to every star $b \in B$ such that $\theta(a,b) < \rho_{max}$ where $\rho_{max}$ is a given threshold.
Hence, let us define
$$P_a = \{ \emptyset \} \cup \{ \{b\} : \theta(a,b) < \rho_{max},~ b \in B\}.$$
Note that the set $\emptyset$ is added to $P_a$ since it could happen that a star of catalog $A$ has no counterpart in $B$.

Let $a \in A \backslash A_2$ and $\{b\} \in P_a$, with its corresponding values
$\alpha^*_a$, $\delta_a$, $m_a$, $\sigma_{\alpha^*_a}$, $\sigma_{\delta_a}$, $\sigma_{m_a}$ and
$\alpha^*_b$, $\delta_b$, $m_b$, $\sigma_{\alpha^*_b}$, $\sigma_{\delta_b}$, $\sigma_{m_b}$ respectively.
A way to measure the probability that $a$ and $b$ are the same star is through the distribution of the 3-dimensional
random vector $(\overline{\alpha^*_a}-\overline{\alpha^*_b}, \overline{\delta_a}-\overline{\delta_b}, \overline{m_a}-\overline{m_b})$,
which is known that it behaves as a multivariate normal distribution whose probability density function is
\[ PDF_1(x,y,z;a,b) = pdf_{dif}(x; \alpha^*, a, b).pdf_{dif}(x; \delta, a, b).pdf_{dif}(x; m, a, b) \]
where
\begin{align*}
  pdf_{dif}(x; \tau, a, b) = pdf(x; \tau_a - \tau_b, \sigma_{\tau_a}^2 + \sigma_{\tau_b}^2), & &\tau \in \{\alpha^*, \delta, m\}
\end{align*}
and $pdf(x; \mu, \sigma) = \frac{1}{\sigma \sqrt{2 \pi}} e^{-\frac{(x-\mu)^2}{2 \sigma^2}}$ is the well known probability density
function of $\mathcal{N}(\mu, \sigma^2)$. Now, define the probability that $a$ corresponds to some $j \in P_a$ as follows:
$$p(a \rightarrow j) = \begin{cases}
p_{\emptyset}, & j = \emptyset, \\
(1 - p_{\emptyset}).\dfrac{PDF_1(0,0,0;a,b)}{\sum_{\{b'\} \in P_a} PDF_1(0,0,0;a,b')}, & j = \{b\},
\end{cases}$$
where $p_{\emptyset}$ is an estimate of the probability that a star from $A$ does not have counterpart in $B$ (usually very low).\\

This treatment generalizes the criterion based on the ``normalized distance''\footnote{See, for instance, W. Sutherland and W. Saunders,
\emph{On the likelihood ratio for source identification}, Mon. Not. R. Astron. Soc. \textbf{259} (1992), 413--420.}
for assigning stars from $A$ to $B$, that is to assign $a \in A$ and $b \in B$ in a way that
$$ND(a,b) \doteq \sqrt{\bigl((\alpha^*_a - \alpha^*_b)/\sigma_{\alpha^*}\bigr)^2 + \bigl((\delta_a - \delta_b)/\sigma_{\delta}\bigr)^2}$$
is minimized, where $\sigma_{\alpha^*}$ and $\sigma_{\delta}$ are the lengths of the axes of the error ellipse:
\begin{lemma}
If $\K = 1$, $|B| \geq |A|$, $p_{\emptyset} \in \mathbb{R}_+$ is almost zero, $\rho_{max} = 180^\circ$,
$\sigma_{\alpha^*} = \sigma_{\alpha^*_a}^2 + \sigma_{\alpha^*_b}^2$ and
$\sigma_{\delta} = \sigma_{\delta_a}^2 + \sigma_{\delta_b}^2$ for all $a \in A$ and $b \in B$, visual magnitudes are not considered
(\ie $m_a = m_b = 0$ and $\sigma_{m_a} = \sigma_{m_b} = 1$ for all $a \in A$ and $b \in B$) and $f^*$ is an optimal assignment
then $f^*$ is a minimum of $ND(f) \doteq \sum_{a \in A} ND(a,f(a))$.
\end{lemma}
\begin{proof}
Note that, for each $a \in A$, $P_a = \{ \emptyset \} \cup \{ \{b\} : b \in B\}$ since $\theta(a,b) < \rho_{max}$ for all $a \in A$ and
$b \in B$. Let $\beta = (1 - p_{\emptyset})/\sum_{b' \in B} PDF_1(0,0,0;a,b')$. Then,
\[ p(a \rightarrow \{b\}) = \beta PDF_1(0,0,0;a,b) =
\beta \frac{1}{\sigma_{\alpha^*} \sqrt{2 \pi}} e^{-\frac{(\alpha^*_a - \alpha^*_b)^2}{2 \sigma_{\alpha^*}^2}}
\frac{1}{\sigma_{\delta} \sqrt{2 \pi}} e^{-\frac{(\delta_a - \delta_b)^2}{2 \sigma_{\delta}^2}} \frac{1}{\sqrt{2 \pi}}. \]
The hypothesis asserts that $p_{\emptyset}$ is small enough to satisfy $p(a \rightarrow \{b\}) > p_{\emptyset}$ for all $b \in B$. 
Let $f$ be a valid assignment. W.l.o.g., suppose that $f(a) \neq \emptyset$ for all $a \in A$. Then,
\begin{multline*}
  w(f) = \sum_{a \in A} w_{a f(a)} = - \sum_{a \in A} ln(p(a \rightarrow f(a)) =\\
    - \dfrac{\beta |A|}{\sigma_{\alpha^*} \sigma_{\delta} (\sqrt{2 \pi})^3}
- \sum_{a \in A} \biggl(-\frac{(\alpha^*_a - \alpha^*_{f(a)})^2}{2 \sigma_{\alpha^*}^2} -\frac{(\delta_a - \delta_{f(a)})^2}{2 \sigma_{\delta}^2} \biggr) =\\
    - \dfrac{\beta |A|}{\sigma_{\alpha^*} \sigma_{\delta} (\sqrt{2 \pi})^3}
+ \dfrac{1}{2} \sum_{a \in A} \biggl(\frac{(\alpha^*_a - \alpha^*_{f(a)})^2}{\sigma_{\alpha^*}^2} + \frac{(\delta_a - \delta_{f(a)})^2}{\sigma_{\delta}^2} \biggr)
\end{multline*}
If $f^*$ is an assignment that minimizes the function $w$, it also minimizes $ND$.
\end{proof}

\medskip

\noindent \emph{Treatment of double stars}.
Let $a \in A_2$ and $\{b_1, b_2\} \in P_a$ (as in the case of single stars, $P_a$ must be obtained with an astrometric criterion
such as the one presented in \cite{journal:mipaper}), with its corresponding values $\alpha^*_a$, $\delta_a$, $m_a$, $\sigma_{\alpha^*_a}$,
$\sigma_{\delta_a}$, $\sigma_{m_a}$, $\alpha^*_{b_1}$, $\delta_{b_1}$, $m_{b_1}$, $\sigma_{\alpha^*_{b_1}}$, $\sigma_{\delta_{b_1}}$,
$\sigma_{m_{b_1}}$, $\alpha^*_{b_2}$, $\delta_{b_2}$, $m_{b_2}$, $\sigma_{\alpha^*_{b_2}}$, $\sigma_{\delta_{b_2}}$, $\sigma_{m_{b_2}}$,
and such that $m_{b_1} < m_{b_2}$, \ie $b_1$ is brighter than $b_2$.
The way to compute the probability that $a$ corresponds to a candidate pair $\{b_1, b_2\}$ highly depends on what is meant by
``double star'' in catalog $A$. In our approach, two features are considered: the angular separation $\theta(b_1, b_2)$ and
the difference in magnitude $m_{b_1} - m_{b_2}$. Let us assume that both features are independent and normally distributed, the first one as
$\mathcal{N}(d_{sep}, \sigma_{d_{sep}})$ and the second one $\mathcal{N}(0,\sigma_{mag})$, where $d_{sep}$, $\sigma_{d_{sep}}$ and
$\sigma_{mag}$ are extra parameters.
Then, the probability of a pair $\{b_1, b_2\}$ is a ``candidate'' is given by a bidimensional random vector whose first component is
the difference between $\theta(b_1,b_2)$ and $d_{sep}$, and the second component is the difference in magnitude $m_{b_1} - m_{b_2}$. Now,
the probability that $a$ corresponds to $\{b_1, b_2\}$ is given by the probability that $a$ and $b_1$ are the same star and $\{b_1,b_2\}$ is
a candidate pair. The following formula defines the probability density function of a 5-dimensional random vector that comprises all together:
\begin{multline*}
PDF_2(x,y,z,w,t;a,b_1,b_2) = PDF_1(x,y,z;a,b_1). \\
pdf(w; \theta(b_1,b_2) - d_{sep}, \sigma_{\theta(b_1,b_2)}^2 + \sigma_{d_{sep}}^2).
pdf(t; m_{b_1} - m_{b_2}, \sigma_{m_{b_1}}^2 + \sigma_{m_{b_2}}^2 + \sigma_{mag}^2)
\end{multline*}
where $\theta(b_1,b_2)$ and $\sigma_{\theta(b_1,b_2)}$ can be computed from position and standard errors of $b_1$ and $b_2$.
Now, define the probability that $a$ corresponds to some $j \in P_a$ as follows:
\[
p(a \rightarrow j) = \begin{cases}
p_{\emptyset}, & j = \emptyset, \\
(1 - p_{\emptyset}).p_{sgl}.\dfrac{PDF_1(0,0,0;a,b)}{\sum_{\{b'\} \in P_a} PDF_1(0,0,0;a,b')},
                 & j = \{b\}, \\
(1 - p_{\emptyset}).(1 - p_{sgl}).\dfrac{PDF_2(0,0,0,0,0;a,b_1,b_2)}{\sum_{\{b'_1,b'_2\} \in P_a} PDF_2(0,0,0,0,0;a,b'_1,b'_2)},
                 & j = \{b_1,b_2\},
\end{cases} \]
where $p_{sgl}$ is an estimate of the probability that a double star from $A$ may be assigned to some single star in $B$.

\medskip

\noindent \emph{Preprocessing catalogs}. The resolution given in Section \ref{sec:model} consists of the cross-identification performed between two known stellar catalogs. The former one is a part of CD (catalog I/114 of VizieR astronomical database) consisting of 52692 stars
whose declinations are between $-22^\circ$ and $-25^\circ$ for epoch $B1875.0$.
The reason for taking these subset of stars is that the information about double stars, \ie the set $A_2$, is only available in printed
form and must be entered by hand. In our case, 571 stars were transcribed, corresponding to the first 177 pages of the printed catalog.

The other catalog is a part of PPMX (catalog I/312 of VizieR) with 130664 stars which cover the sky region of the former one.

The preprocessing of both catalogs is essentially the same as in \cite{journal:mipaper}. 
Some stars from CD have been deliberately removed due to the following causes: 1) variable star; 2) cumulus; 3) a star appearing in
PPM (catalogs I/193, I/206 and I/208 of VizieR) and whose position in PPM differs from CD in more than 2 arcmin for epoch $B1875.0$ or
whose magnitude differs from CD in more than $1.5$. Some other entries in catalog CD has been altered because of typo errors
\cite{journal:mipaper}. After this process, there are 52313 stars left (where 568 are doubles).

Data from PPMX catalog have been preprocessed as follows. Visual magnitudes have been converted to the magnitude scale used by CD:
$m_{CD} = -0.01335368 m^2+ 1.076636 m + 0.2249828$
where $m$ is the Johnson V magnitude reported in PPMX and $m_{CD}$ is the target
magnitude. These coefficients have been obtained through a quadratic fit explained in \cite{journal:mipaper}.
Positions have been translated to the epoch of CD. In addition, the column of visual magnitude (specifically, Johnson V) for several entries
of PPMX is empty so it has been filled with magnitudes from catalog APASS-DR9 (catalog II/336 of VizieR). After this process, stars with magnitude greater than $13.5$ have been discarded, leaving 83397 stars.

A preliminary cross-identification between CD and PPMX has been performed via the X-Match Service in order to generate the sets of candidate
stars $P_a$ faster. The parameters and standard errors have been set as follows:
\begin{itemize}
\item $\rho_{max} = 2$~arcmin~~(the maximum allowed by X-Match)
\item $d_{sep} = 34.9$~arcsec~~\cite{journal:mipaper}
\item $\sigma_{d_{sep}} = 13.65$~arcsec~~\cite{journal:mipaper}
\item $\sigma_{mag} = 0.915$~~\cite{journal:mipaper}
\item $p_{\emptyset} = 10^{-10}$
\item $p_{sgl} = 10^{-4}$
\item $\sigma_{\alpha^*_a}^2 + \sigma_{\alpha^*_b}^2 = (10.15$~arcsec$)^2$,
      $\sigma_{\delta_a}^2 + \sigma_{\delta_b}^2 = (22.74$~arcsec$)^2$,
      $\sigma_{m_a}^2 + \sigma_{m_b}^2 = 0.2759^2$ for all $a \in A$ and $b \in B$~~\cite{journal:mipaper}
\item $\sigma_{\theta(b_1,b_2)}^2 = 0$, $\sigma_{m_{b_1}}^2 + \sigma_{m_{b_2}}^2 = 0$ for all $b_1, b_2 \in B$
\end{itemize}
The dataset \cite{dataset:mendeley} contains the new CD catalog with the cross-identification (\texttt{new\_cd.txt}), its format
(\texttt{new\_format.txt}), the source code as well as other auxiliary files.
In Figure \ref{fig:diagram} a picture of the whole process is displayed.

\medskip

\noindent \emph{Acknowledgements.} I would like to thank Mar\'ia Julia Sever\'in and the people mentioned in the Acknowledgements of
\cite{journal:mipaper} who help me to enter the set of double stars, among other data.

\begin{figure}[b]
	\centering
		\includegraphics[scale=0.4]{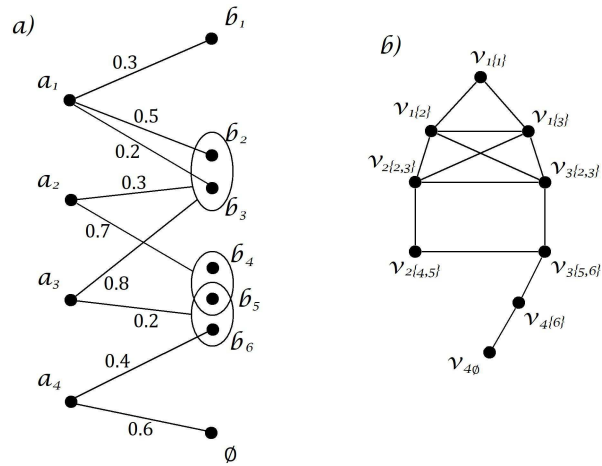}
	\vspace{-20pt}
	\caption{Example of 2-Matching Problem: a) instance, b) $G$}
	\label{fig:ex1}
\end{figure}

\begin{figure}[b]
	\centering
		\includegraphics[scale=0.6]{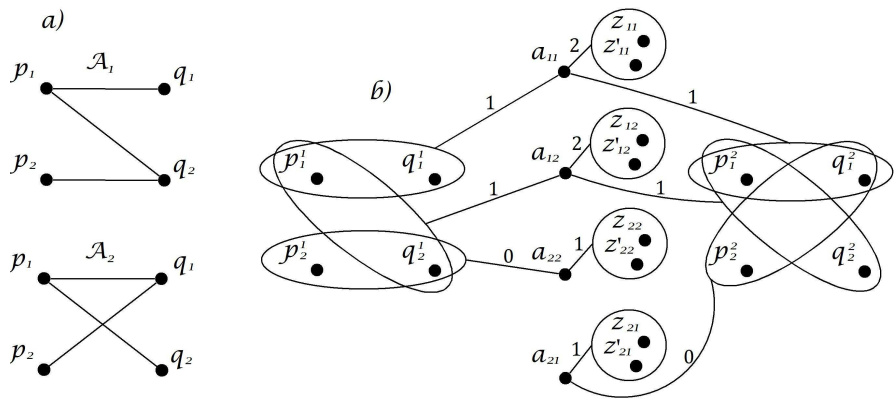}
	\caption{Example of reduction: a) DM, b) 2-MDPAW}
	\label{fig:proof1}
\end{figure}

\begin{figure}
	\centering
		\includegraphics[scale=0.6]{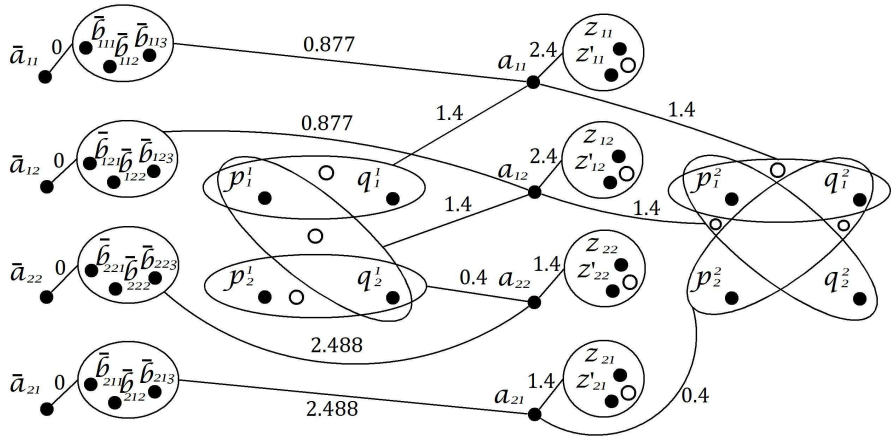}
	\caption{Example of reduction to 3-MDP}
	\label{fig:proof2}
\end{figure}

\begin{center}
\renewcommand{\arraystretch}{1.0}
\begin{table}
\begin{tabular}{|l|c|}
\hline
Number of stars of catalog $A$ ($n_A$) & 52313 \\
\hline
Number of stars of catalog $B$ ($n_B$) & 83397 \\
\hline
Double stars present in catalog $A$ ($|A_2|$) & 568 \\
\hline
Largest cardinal of $P_a$ & 34 \\
\hline
Number of components of $G$: & \\
$\bullet$ Unique star & 39383 \\
$\bullet$ Only single stars & 5628 \\
$\bullet$ Multiple stars & 111 \\
\hline
Largest cardinal of $A'$ found in components: & \\
$\bullet$ Only single stars & 34 \\
$\bullet$ Multiple stars & 7 \\
\hline
Statistics of the solution: & \\
$\bullet$ Unassigned stars & 245 \\
$\bullet$ Single stars assigned & 51502 \\
$\bullet$ Double stars assigned & 483 \\
$\bullet$ Double stars assigned to a single one in $B$ & 83 \\
\hline
\end{tabular}
\vspace{10pt}
\caption{Highlights about the optimization} \label{table:optim}
\end{table}
\end{center}

\begin{figure}
	\centering
		\includegraphics[scale=0.8]{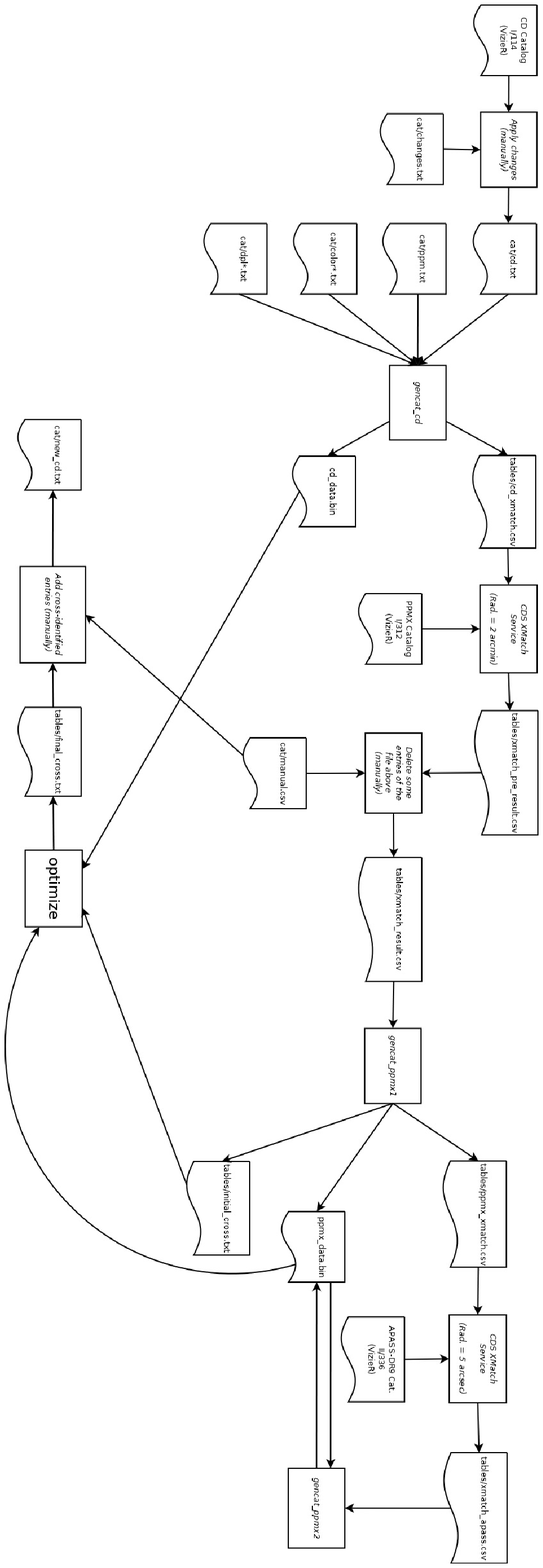}
	\caption{Diagram}
	\label{fig:diagram}
\end{figure}

\end{document}